\documentclass[submission,copyright,creativecommons]{eptcs}
\providecommand{\event}{ITRS 2016} 

\usepackage{microtype}

\usepackage{amsmath, amsfonts, amsthm, amssymb, mathtools} 
\usepackage{proof}              
\usepackage{mathpartir}         
\usepackage{cleveref}           

\newcommand\indexVar{k}
\newcommand\lab{lab}

\usepackage{macro/generic}
\usepackage[index=\indexVar,label=\lab]{macro/language}
\usepackage{macro/code}

\renewcommand{\intersect}{\mathrel{\sqcap}}
\renewcommand{\union}{\mathrel{\sqcup}}
\newcommand{\m}[1]{\mathsf{#1}}
\newcommand{\mb}[1]{\mbox{\bf #1}}
\newcommand{\semi}{\mathrel{;}}

\title{Intersections and Unions of Session Types}
\author{Co\c{s}ku Acay
\institute{Carnegie Mellon University\\
Pennsylvania, USA}
\email{cacay@cmu.edu}
\and
Frank Pfenning
\institute{Carnegie Mellon University \\
Pennsylvania, USA}
\email{fp@cs.cmu.edu}
}
\def\titlerunning{Intersections and Unions of Session Types}
\def\authorrunning{C. Acay \& F. Pfenning}
\begin{document}
\maketitle

\begin{abstract}
  Prior work has extended the deep, logical connection between the linear sequent calculus and session-typed message-passing concurrent computation with equi-recursive types and a natural notion of subtyping. In this paper, we extend this further by intersection and union types in order to express multiple behavioral properties of processes in a single type. We prove session fidelity and absence of deadlock and illustrate the expressive power of our system with some simple examples. We observe that we can represent internal and external choice by intersection and union, respectively, which was previously suggested in \cite{CastagnaDGP09,Padovani10} for a different language of session types motivated by operational rather than logical concerns.
\end{abstract}

\section{Introduction}

Prior work has established a Curry-Howard correspondence between intuitionistic linear sequent calculus and session-typed message-passing concurrency \cite{CairesP10, PfenningG15, Honda93}. In this formulation, linear propositions are interpreted as session types, proofs as processes, and cut reduction as communication. Session types are assigned to channels and prescribe the communication behavior along them. Each channel is offered by a unique process and used by exactly one, which is ensured by linearity. When the behavior along a channel $c$ satisfies the type $A$ and $P$ is the process that offers along $c$, we say that $P$ provides a session of type $A$ along $c$.

In the base system, each type directly corresponds to a process of a certain form. For example, a process providing the type $A \tensor B$ first sends out a channel satisfying $A$, then acts as $B$. Similarly, a process offering $\terminate$ sends the label $\irb{end}$ and terminates. We call these \emph{structural types} since they correspond to processes of a certain structure. In this paper, we extend the base type system with intersections and unions. We call these \emph{property types} since they do not correspond to specific forms of processes in that any process may be assigned such a type. In addition, if we interpret a type as specifying a property, then intersection corresponds to satisfying two properties simultaneously and union corresponds to satisfying one or the other.

Our goal is to show that the base system extended with intersection, unions, recursive types, and a natural notion of subtyping is type-safe. We do this by proving the usual type preservation and progress theorems, which correspond to session fidelity and deadlock freedom in the concurrent context. In the presence of a strong subtyping relation and transparent (i.e.\ non-generative)  equi-recursive types, intersections and unions turn out to be powerful enough to specify many interesting communications behaviors, which we demonstrate with examples analogous to those in functional languages \cite{FreemanP91,Dunfield03}.

Our contributions are summarized below:
\begin{itemize}
  \item We introduce intersection and union types to a session-typed concurrent calculus and prove session fidelity and deadlock freedom.
  \item We give a simple and sound coinductive subtyping relation in the presence of equi-recursive types, intersections, and unions reminiscent of Gentzen's multiple conclusion sequent calculus \cite{Gentzen35, Girard87}.
  \item We show how intersections and unions can be used as refinements of recursive types in a linear setting.
  \item We show decidability of subtyping and present a system for algorithmic type checking.
  \item We demonstrate how internal and external choice can be understood as singletons interacting with intersection and union.
\end{itemize}

An extended version of this paper can be found at \cite{Acay16}.

\section{From Linear Logic to Session Types}
\label{base}
We give only a brief review of linear logic and its connection to session types here. Interested readers are referred to \cite{CairesP10, PfenningG15, Honda93}. The key idea of linear logic is to treat logical propositions as resources: each must be used exactly once in a proof. According to the Curry-Howard isomorphism for intuitionistic linear logic, propositions are interpreted as session types, proofs as concurrent processes, and cut-elimination steps as communication. For this correspondence, hypotheses are labelled with channels (rather than with variables). We also assign a channel name to the conclusion so that a process \emph{providing} a session can refer to the same channel name as process \emph{using} it. This replaces the usual notion of \emph{duality} in classical session-typed calculi~\cite{Honda93} and gives us the following form for typing judgments:
$$ \typeD {c_1 : A_1, \ldots, c_n : A_n} P c A$$
which should be interpreted as ``\emph{$P$ provides along the channel $c$ the session $A$ using channels $c_1, \ldots, c_n$ (linearly) with their corresponding types}''. We assume $c_1, \ldots, c_n$ and $c$ are all distinct so that when a channel $c$ occurs in $P$, its reference is unambiguous.  We abbreviate hypotheses using $\ctx$,
$\ctx'$, etc.

Two key rules explaining this judgment are \emph{cut} and \emph{identity}.  Logically, cut means that if we can prove $A$, then we can use it as a resource in the proof of some other proposition $C$.  Operationally, it corresponds to \emph{parallel composition} with a private shared variable for communication.
$$ \infer[\cut]{ \typeD {\ctx, \ctx'} {\tspawn{c}{P_c}{Q_c}} {d} {D} }
    { \typeD {\ctx} {P_c} {c} {A}
    & \typeD {\ctx', c : A} {Q_c} {d} {D}
    }
$$
We write $\tspawn{c}{P_c}{Q_c}$ instead of the customary $(\nu c)(P_c \mid Q_c)$ because it is more readable in actual programs.  The subscript here indicates that $c$ is a bound variable and occurs in both $P_c$ and $Q_c$.

Logically, the identity states the fundamental principle that the resource $A$ can be used to prove the conclusion $A$.  Operationally, it corresponds to \emph{forwarding}: we provide a session along channel $c$ by forwarding to $d$.  Note that forwarding does not have a direct correspondence in the $\pi$-calculus but can be implemented in terms of other primitives (for example, $!(c(d).\overline{c}\langle d\rangle + d(c).\overline{d}\langle c\rangle)$).
$$
  \infer[\id]{ \typeD {d : A} {\tfwd{c}{d}} {c} {A} }
    {}
$$
Forwarding is employed in session-typed programming with surprising frequency as will be evident from the examples we provide.


Cut and identity are general constructs, independent of any particular proposition.  The isomorphism takes shape once we work out the interpretations of all of the connectives of linear logic as \emph{session types}.  We foreshadow the operational interpretation from the perspective of the provider of a session:

\begin{center}
\begin{tabular}{l c l l}
  $A, B, C$ & ::= & $\terminate$        & send \texttt{end} and terminate \\
            & $|$ & $A \tensor B$       & send channel of type $A$ and continue as $B$ \\
            & $|$ & $\internals{A}{I}$  & send $\lab_i$ and continue as $A_i$ for some $i \in I$\\
            & $|$ & $A \lolli B$        & receive channel of type $A$ and continue as $B$ \\
            & $|$ & $\externals{A}{I}$  & receive $\lab_i$ and continue as $A_i$ for some $i \in I$
\end{tabular}
\end{center}

In this paper, we do not need ${!}A$.  Instead of replication, we use recursive types to describe recurring behavior. We also generalize the binary $A \oplus B$ to $\internals{A}{I}$ and $A \& B$ to $\externals{A}{I}$, which is in the tradition of much prior work on session types and makes programs more readable. Here, $I$ is a \emph{finite} non-empty set of distinct labels whose order does not matter. We call $\internals{A}{I}$ an internal choice since the provider selects which branch to take (i.e.\ the provider sends the label), and $\externals{A}{I}$ an external choice since the client makes the decision.

\subsection{Process Expressions}
\label{process-expressions}

The processes (or proof terms) corresponding to these types are given below with the sending construct followed by the receiving construct. The notation $P_x$ emphasizes the scope of a bound variable $x$ in the cut, send, and receive constructs.  As a shorthand for substitution, we write $P_a$ for $\subst a x P$, the capture-avoiding substitution of $a$ for $x$ in $P$.
\begin{center}
\begin{tabular}{l c l@{\hspace{2em}} l}
  $P, Q, R$ & ::= & $\tspawn{x}{P_x}{Q_x}$     & cut (spawn) \\
            & $|$ & $\tfwd c d$                & id (forward) \\
            & $|$ & $\tclose c \mid \twait c P$  & $\terminate$ \\
            & $|$ & $\tsend{c}{y}{P_y}{Q} \mid \trecv{x}{c}{R_x}$ & $A \tensor B,$ $A \lolli B$ \\
            & $|$ & $\tselect{c}{}{P} \mid \tcase c {\tbranches Q I}$  & $\externals A I,$ $\internals A I$
\end{tabular}
\end{center}

\par Because a process $P$ always provides a session along a unique channel, we sometimes identify a process with the channel along which it provides. Keeping this in mind, the intuitive readings of process expressions are as follows. $\tspawn{x}{P_x}{Q_x}$ creates a fresh channel $c$, spawns a new process $P_c$ providing along the channel $c$, and $Q_c$ uses this new channel. The forwarding process $\tfwd c d$ terminates by globally identifying channels $c$ and $d$.  Further communication along $c$ will instead take place along $d$. A possible implementation of a forwarding process might tell the two end-points to communicate with each other and itself terminate, effectively, tying the two ends together and stepping out of the way. Process $\tclose c$ sends the token \texttt{end} and terminates, whereas the matching $\twait c P$ waits for \texttt{end} along $c$ and continues as $P$.  Processes $\mathrm{send}$ and $\mathrm{recv}$ are used for communicating channels along channels. In the case of $\tsend c y {P_y} Q$, a new channel $d$ is created and a process $P_d$ is spawned, but unlike cut, the continuation ($Q$) cannot refer to the new channel $d$, since it is sent along $c$ to be used by a different process. Finally, $\tselect c {} P$ sends $\lab$ along $c$, and $\tcase c {\tbranches Q I}$ receives a label along $c$ and branches on it. These are used to determine which branches to take in internal and external choices.

\subsection{Recursively Defined Types and Processes}

For practical programming, we of course need recursive types and recursively defined processes. Usually, this is incorporated into the type language in a local way by introducing a new construct $\recursive{t}{A_t}$ and identifying $\recursive{t}{A}$ with its unfolding $\subst{\recursive{t}{A}}{t}{A}$ in the style of Amadio and Cardelli~\cite{AmadioC91}. This makes it harder to incorporate mutual recursion, however, requiring type level pairs~\cite{Stone05un} for a fully formal treatment. We therefore go with the much simpler approach of using a global signature $\recCtx$ of mutually recursive \emph{type definitions} and \emph{process definitions}. Type definitions are straightforward with the form $t = A$. They add defined type names as a new alternative to possible types. The fact that type definitions are transparent, that is, the fact that they are identified with their definition (or unfolding) without process-level terms to witness the isomorphism, corresponds to an equi-recursive interpretation of types. Process definitions have the form $X :: (c:A) = P_c$ where $P$ provides session $A$ along $c$.
$$ \recCtx ::= \cdot \mid \recCtx, t = A \mid X ::(c:A) = P_c $$
Note that $c$ is a bound variable with scope $P$. Formally, our judgments are now decorated with a fixed signature $\recCtx$, but we might elide it since the signature usually does not change during a derivation.

The identification of defined type names and their unfoldings is normally formalized by defining a new type equivalence judgment and adding conversion rules to process typing. In our formalization, we will integrate this into the subtyping judgment, but we delay its discussion to \cref{original-subtyping}. For now, it suffices to think of type definitions as finite representations of possibly infinite trees. For example, the definition $t = t \tensor t$ stands for the tree $(\ldots \tensor \ldots) \tensor (\ldots \tensor \ldots)$. However, for this to make sense, we need to slightly restrict valid type definitions to those that are contractive~\cite{Stone05un, GayH05}. In our setting, this means every type definition must have a structural type at the top. For example, the previous example of $t = t \tensor t$ is contractive since its top level construct is $\tensor$ whereas $t = t$ and $t = u$ are not.

As an example for signatures with recursion, consider processes producing a sequence of the label $\m{succ}$ followed a single $\m{zero}$ (i.e.\ Peano naturals).
\[
\begin{array}{lcl}
\m{Nat} & = & \oplus\{\m{zero} : \terminate, \m{succ} : \m{Nat}\} \\
\m{z} :: (c:\m{Nat}) & = & c.\m{zero} \semi \mb{close}\; c \\
\m{s} :: (c:\m{Nat} \lolli \m{Nat}) & = & d \leftarrow \mb{recv}\; c \semi c.\m{succ} \semi c \leftarrow d
\end{array}
\]
The $\m{s}$ process here receives a channel $d$ representing a natural number, outputs one $\m{succ}$, and then behaves like $d$.  Effectively, this computes the successor.  To double a natural number, we need a recursive process definition.
\[
\begin{array}{l}
\m{double} :: (c:\m{Nat} \lolli \m{Nat}) = \\
\quad d \leftarrow \mb{recv}\; c \semi \\
\quad \mb{case}\; d\; \mb{of} \\
\quad\quad \m{zero} \rightarrow \mb{wait}\; d \semi c.\m{zero} \semi \mb{close}\; c \\
\quad\quad \m{succ} \rightarrow c.\m{succ} \semi c.\m{succ} \semi \\
\hspace{6em} e \leftarrow \m{double} \semi \mb{send}\; e\; (d' \leftarrow (d' \leftarrow d)) \semi c \leftarrow e
\end{array}
\]

The last line here implements the following sequence:
\[
\begin{array}{ll}
e \leftarrow \m{double} & \mbox{cut: start a new process $\m{double}$ which provides along the new channel $e$} \\
\mb{send}\; e\; (d' \leftarrow (d' \leftarrow d)) & \mbox{essentially, send $d$ along $e$ to the new process, but we have to use a forward} \\
c \leftarrow e & \mbox{now forward $e$ to $c$}
\end{array}
\]
This is a frequent pattern, so we make two small improvements in the surface syntax: (1) we parameterize process definitions with the channels that they provide and use, and (2) instead of creating a new process and then forwarding it, we directly implement the provider channel with a new call. Then the definition becomes the much more readable
\[
\begin{array}{l}
\m{double} : \m{Nat} \lolli \m{Nat} \\
c \leftarrow \m{double}\; d = \\
\quad \mb{case}\; d\; \mb{of} \\
\quad\quad \m{zero} \rightarrow \mb{wait}\; d \semi c.\m{zero} \semi \mb{close}\; c \\
\quad\quad \m{succ} \rightarrow c.\m{succ} \semi c.\m{succ} \semi \\
\hspace{6em} c \leftarrow \m{double}\; d
\end{array}
\]
In an actual implementation, the new definition would be ``desugared'' to the previous one.

\subsection{Type Assignment for Processes}

The typing rules for processes are derived from linear logic by decorating derivations with proof terms. The rules are given in \cref{session-assignment}. Note that in $\internal\Left$ and $\external\Right$, we allow unused branches in case expressions. This makes width subtyping easier, which is discussed in \cref{original-subtyping}. In addition, the $\m{def}$ rule implicitly renames the channel name in the signature to the one expected by the judgment.

In our simple language, checking that a type is valid, $\vdash_\recCtx A : \m{type}$, just verifies that all type names in $A$ are defined in $\recCtx$ and that $A$ is contractive.  A signature $\recCtx$ itself is checked with the rules below.  Note that we allow mutual recursion in the definitions which is witnessed by the fact that the signature $\recCtx$ is propagated identically everywhere.
\[
\infer[]{\vdash_\recCtx \emptyset}{}
\hspace{3em}
\infer[]{\vdash_\recCtx \recCtx', t = A}{\vdash_\recCtx \recCtx' &
\vdash_\recCtx A : \m{type}}
\hspace{3em}
\infer[]{\vdash_\recCtx \recCtx', X :: (c:A) = P_c}
{\vdash_\recCtx \recCtx' & \vdash_\recCtx A : \m{type} & \emptyset \vdash_\recCtx {P_c} :: (c:A)}
\]

\begin{rules}[session-assignment]{Type assignment for process expressions}
  \infer[\id]{ \typeD {c : A} {\tfwd{d}{c}} {d} {A} }
    {}
  \and \infer[\cut]{ \typeD {\ctx, \ctx'} {\tspawn{c}{P_c}{Q_c}} {d} {D} }
    { \typeD {\ctx} {P_c} {c} {A}
    & \typeD {\ctx', c : A} {Q_c} {d} {D}
    }
  \and \infer[\terminate\Right]{\typeD{\emptyCtx}{\tclose c}{c}{\terminate}}
    {}
  \and \infer[\terminate\Left]{\typeD{\ctx, c : \terminate}{\twait c P}{d}{A}}
    {\typeDJ{P}{d}{A}}
  \and \infer[\tensor\Right]{\typeD{\ctx, \ctx'}{\tsend{c}{d}{P_d}{Q}}{c}{A \tensor B}}
    { \typeD{\ctx}{P_d}{d}{A}
    & \typeD{\ctx'}{Q}{c}{B}
    }
  \and \infer[\tensor\Left]{ \typeD{\ctx, c : A \tensor B}{\trecv{d}{c}{P_d}}{e}{E} }
    { \typeD{\ctx, d : A, c : B}{P_d}{e}{E} }
  \and \infer[\internal\Right]{\typeDJ { \tselect{c}{i}{P} } {c} {\internals{A}{I}} }
    { i \in I
    & \typeDJ{P}{c}{A_i}
    }
  \and \infer[\internal\Left]{ \typeD { \ctx, c : \internals{A}{I} } { \tcase{c}{\tbranches{P}{J}} } {d} {D} }
   { I \subseteq J
   & \typeD{\ctx, c : A_k}{P_k}{d}{D}~\text{for}~k\in I
   }
  \and \infer[\lolli\Right]{ \typeD{\ctx}{\trecv{d}{c}{P_d}}{c}{A \lolli B} }
    { \typeD{\ctx, d : A}{P_d}{c}{B} }
  \and \infer[\lolli\Left]{\typeD{\ctx, \ctx', c : A \lolli B}{ \tsend{c}{d}{P_d}{Q} } {e}{E}}
    { \typeD{\ctx}{P_d}{d}{A}
    & \typeD{\ctx', c : B}{Q}{e}{E}
    }
  \and \infer[\external\Right]{ \typeDJ { \tcase{c}{\tbranches{P}{I}} } {c} {\externals{A}{J}} }
   { J \subseteq I
   & \typeDJ{P_k}{c}{A_k}~\text{for}~k\in J
   }
  \and \infer[\external\Left]{\typeD{\ctx, c : \externals{A}{I}} { \tselect{c}{i}{P} } {d} {D}}
    { i \in I
    & \typeD{\ctx, c : A_i}{P}{d}{D}
    }
  \and \infer[\m{def}]{ \typeD {\emptyset} {X} {d}{A} }
    { X = P :: (c : A) \in \recCtx }
\end{rules}

\subsection{Subtyping}
\label{original-subtyping}

Gay and Hole \cite{GayH05} add coinductive subtyping (denoted $A \sub B$ in this paper) to their system in order to admit width and depth subtyping for $n$-ary choices, which are standard for record-like and variant-like structures. In our system, subtyping also doubles as a convenient way of identifying a recursive type and its unfolding using the following rules:
\begin{mathpar}
  \infer=[\m{Def}\Right]{A \sub_\recCtx t}
    {(t = B \in \recCtx) & A \sub_\recCtx B}
  \and \infer=[\m{Def}\Left]{t \sub_\recCtx B}
    {(t = A \in \recCtx) & A \sub_\recCtx B}
\end{mathpar}
Double lines here indicate that the rules should be interpreted coinductively as is common with theories using equi-recursive types. We will not go into the details of Gay and Hole's system since we will switch to a different relation in the next section anyway. Either way, we relate subtyping to process typing with subsumption rules:
\begin{mathpar}
  \infer[\irb{Sub}\Right]{\typeRecDJ{P}{c}{A}}
    {\typeRecDJ{P}{c}{A'} & A' \sub_\recCtx A}
  \and \infer[\irb{Sub}\Left]{\typeRecD {\ctx, c : A} \recCtx P d B}
    {\typeRecD{\ctx, c : A'} \recCtx P d B & A \sub_\recCtx A'}
\end{mathpar}

\subsection{Process Configurations}

So far in the theory, we have only considered processes in isolation. In this section, we introduce process configurations in order to talk about the interactions between multiple processes. A process configuration, denoted by $\config$, is simply a set of processes where each process is labelled with the channel along which it provides. We use the notation $\proc c P$ for labelling the process $P$, and require all labels in a configuration to be distinct.

With the above restriction, each process offers along a specific channel and each channel is offered by a unique process. Since channels are linear resources in our system, they must be used by exactly one process. In addition, we do not allow cyclic dependence, which imposes an implicit forest (set of trees) structure on a process configuration where each node has one outgoing edge
and any number of incoming edges that correspond to channels the process uses. This observation suggests the typing rules below, which mimic the structure of a multi-way tree. Note that the definition is well founded since the size of the configuration gets strictly smaller.
\begin{mathpar}
  \infer[\confZero]{\providesCtx{\emptyConfig}{\emptyCtx}} {}
  \and \infer[\confOne]{\provides{\config, \proc c P}{c}{A}}
   { \providesCtx \config \ctx
   & \typeD \ctx P c A
   }
   \and \infer[\confN]{\providesCtx {\config_1, \ldots, \config_n} {\parens{c_1 : A_1, \ldots, c_n : A_n}}}
    { \provides {\config_i} {c_i} {A_i}~\text{for}~i \in \set{1, \ldots, n}
    & n > 1
    }
\end{mathpar}

\subsection{Operational Semantics}

A process configuration evolves over time when a process takes a step, either by spawning a new process ($\cut$), forwarding ($\id$) or when two matching processes communicate. Our configurations are sets, so order is not significant when we match the left-hand sides against the configuration $\Omega$.  When we require a new name to be chosen, it must not already be offered by some process in the configuration.  These rules are an example of
a \emph{substructural operational semantics}~\cite{Simmons12}, presented in the form of a \emph{multiset rewriting system}~\cite{CervesatoS09}.

\[
\begin{array}{lcll}
\Omega, \proc{c}{\tfwd{c}{d}} & \longrightarrow & [d/c]\Omega \\
\Omega, \proc{c}{\tspawn{x}{P_x}{Q_x}} & \longrightarrow &
\Omega, \proc{a}{P_a} , \proc{c}{Q_a} & \mbox{($a$ fresh)} \\
\Omega, \proc{c}{X} & \longrightarrow & \multicolumn{2}{r}{\Omega, \proc{c}{[c/d]P} \quad \mbox{($X = P :: (d:A) \in \recCtx$)}} \\
\Omega, \proc{c}{\tclose{c}} , \proc{e}{\twait{c}{P}}
    & \longrightarrow & \Omega, \proc{e}{P} \\
\Omega, \proc{c}{\tsend{c}{x}{P_x}{Q}} , \proc{e}{\trecv{x}{c}{R_x}}
    & \longrightarrow & \\
& \multicolumn{3}{r}{\Omega,  \proc{a}{P_{a}} , \proc{c}{Q} , \proc{e}{R_{a}} \quad \mbox{($a$ fresh)}} \\
\Omega, \proc{c}{\tselect{c}{i}{P}} , \proc{e}{\tcase{c}{\tbranches Q I}}
& \longrightarrow & \Omega, \proc{c}{P} , \proc{e}{Q_i}  & \mbox{($i \in I$)} \\
\Omega, \proc{c}{\trecv{x}{c}{P_x}} , \proc{d}{\tsend{c}{x}{Q_x}{R}} & \longrightarrow & \\
& \multicolumn{3}{r}{\Omega, \proc{c}{P_{a}} , \proc{a}{Q_a} , \proc{d}{R} \quad \mbox{($a$ fresh)}} \\
\Omega, \proc{c}{\tcase{c}{\tbranches P I}} , \proc{e}{\tselect c i Q} 
& \longrightarrow & \Omega, \proc{c}{P_i} , \proc{e}{Q} & \mbox{($i \in I$)} \\
\end{array}
\]

This concludes the discussion of the base system. In the next section, we introduce intersections, unions, and a multiple-conclusion subtyping relation which constitute our main contributions.

\section{Intersections and Unions}
\label{refinements}

Recall our definition of process-level naturals $\nat$. One can imagine cases where we would like to know more about the exact nature of the natural. For example, if we are using a natural to track the size of a list, we might want to ensure it is non-zero. Sometimes, it might be relevant to track whether we have an even or an odd number. The system we have described so far turns out to be strong enough to describe all these \emph{refinements} as illustrated below
\[
\begin{array}{lcll}
  \m{Nat} & = & \internal\{\m{zero} : \terminate, \m{succ} : \m{Nat}\} \\[1em]
  \m{Pos} & = & \internal\{\m{succ} : \m{Nat}\} \\
  \m{Even} & = & \internal\{\m{zero} : \terminate, \m{succ} : \m{Odd}\} \\
  \m{Odd} & = & \internal\{\m{succ} : \m{Even}\} \\[1em]
\end{array}
\]
Recall also the definitions
\[
\begin{array}{l@{\hspace{2em}}l@{\hspace{2em}}l}
  \m{z} : \m{Nat} & \m{s} : \m{Nat} \lolli \m{Nat} & \m{double} : \m{Nat} \lolli \m{Nat} \\
  c \leftarrow \m{z} = & c \leftarrow \m{s}\; d = & c \leftarrow \m{double}\; d = \\
  \quad c.\m{zero} \semi \mb{close}\; c & \quad c.\m{succ} \semi c \leftarrow d & \quad \mb{case}\; d\; \mb{of} \\
& & \quad\quad \m{zero} \rightarrow \mb{wait}\; d \semi c.\m{zero} \semi \mb{close}\; c \\
& & \quad\quad \m{succ} \rightarrow c.\m{succ} \semi c.\m{succ} \semi \\
& & \hspace{6em} c \leftarrow \m{double}\; d
\end{array}
\]


Intuitively, it is easy to see that $\pos$, $\even$, and $\odd$ are all subtypes of $\nat$. We run into a problem when we try to implement the behavior described by these types, however. The $\m{s}$ process, for example, satisfies many properties: $\nat \lolli \nat$, $\pos \lolli \pos$, $\even \lolli \odd$, $\odd \lolli \even$ etc. Subtyping can be used to combine some of these (e.g.\ $\nat \lolli \pos$ for $\nat \lolli \nat$ and $\pos \lolli \pos$) but it is not expressive enough to combine all properties. An elegant solution is to add intersections to the type system.

\subsection{Intersection Types}
We denote the intersection of two types $A$ and $B$ as $A \intersect B$. A process offers an intersection type if its behavior satisfies both types simultaneously. Using intersections, we can assign the programs introduced in \cref{process-expressions} types specifying all behavioral properties we care about:
\[
\begin{array}{lcl}
\m{z} & : & \m{Nat} \intersect \m{Even} \\
\m{s} & : & (\m{Nat} \lolli \m{Nat}) \intersect (\m{Even} \lolli \m{Odd}) \intersect
(\m{Odd} \lolli \m{Even}) \\
\m{double} & : & (\m{Nat} \lolli \m{Nat}) \intersect (\m{Nat} \lolli \m{Even})
\end{array}
\]
Note that as is usual with intersections, multiple types are assigned to \emph{the same process}. Put differently, we cannot use two different processes or specify two different behaviors to satisfy the different branches of an intersection. This leads to the following typing rule: 
$$
  \infer[\intersect\Right]{\typeRecDJ{P}{c}{A \intersect B}}
    {\typeRecDJ{P}{c}{A} & \typeRecDJ{P}{c}{B}}
$$

When we are using a channel on the left that offers an intersection of two types, we know it has to satisfy both properties so we get to pick the one we want:
\begin{mathpar}
  \infer[\intersect\Left_1]{\typeRecD{\ctx, c : A \intersect B}{\recCtx}{P}{d}{D}}
    {\typeRecD{\ctx, c : A}{\recCtx}{P}{d}{D}}
  \and \infer[\intersect\Left_2]{\typeRecD{\ctx, c : A \intersect B}{\recCtx}{P}{d}{D}}
    {\typeRecD{\ctx, c : B}{\recCtx}{P}{d}{D}}
\end{mathpar}

It may seem as if the two left typing rules for intersection are somehow unnecessary: because of linearity, only one of $A$ or $B$ can be selected in any given derivation.  But process definitions are used arbitrarily often, essentially spawning a new process along a new linear channel at each use point, so we may need to select a different component of the type at each occurrence.  For example:

\[
\begin{array}{l}
  \m{s} : (\m{Even} \lolli \m{Odd}) \intersect (\m{Odd} \lolli \m{Even}) \quad \mbox{\it (from before)} \\[1ex]
  \m{s2} : \m{Even} \lolli \m{Even} \\[1ex]
  c \leftarrow \m{s2}\; d = \\
  \quad d_1 \leftarrow \m{s}\; d \hspace{5em} \mbox{\it (use $\m{s} : \m{Even} \lolli \m{Odd}$ by $\intersect\Left_1$)} \\
  \quad c \leftarrow \m{s}\; d_1 \hspace{5em} \mbox{\it (use $\m{s} : \m{Odd} \lolli \m{Even}$ by
$\intersect\Left_2$)}
\end{array}
\]

The standard subtyping rules are given below. It should be noted that the left typing rules above are derivable by an application of subsumption on the left using $\Sub{\intersect\Left_1}$ and $\Sub{\intersect\Left_2}$, so we will not explicitly add these to the final system. Also, we will have to modify the subtyping relation later in this section, so these rules are only a first attempt.

\begin{mathpar}
  \infer=[\Sub{\intersect}\Right]{A \sub B_1 \intersect B_2}
    {A \sub B_1 \and A \sub B_2}
  \and \infer=[\Sub{\intersect\Left_1}]{A_1 \intersect A_2 \sub B}
    {A_1 \sub B}
  \and \infer=[\Sub{\intersect\Left_2}]{A_1 \intersect A_2 \sub B}
    {A_2 \sub B}
\end{mathpar}

One final note about intersection types and recursion is that intersections are not considered structural types and thus do not contribute to contractiveness. That is, the type $t = t \intersect t$ is \emph{not} contractive.

\subsection{Union Types}

Unions are the dual of intersections and correspond to processes that satisfy one or the other property, and are written $A \union B$. We add unions because they are a natural extension to a  type system with intersections. We will also see how $n$-ary internal choice can be interpreted as
the union of singleton choices. Without them, our interpretation would only be half-complete since we could interpret external choice (with intersections) but not internal choice.

Being dual to intersections, the typing rules for unions mirror the typing rules for intersections: we have two right rules and one left rule, and this time the right rules are derivable from subtyping. The rules are given below:

\begin{mathpar}
  \infer[\union\Right_1]{\typeRecDJ{P}{c}{A \union B}}
    {\typeRecDJ{P}{c}{A}}
  \and \infer[\union\Right_2]{\typeRecDJ{P}{c}{A \union B}}
    {\typeRecDJ{P}{c}{B}}
  \and \infer[\union\Left]{\typeRecD{\ctx, c : A \union B}{\recCtx}{P}{d}{D}}
    {\typeRecD{\ctx, c : A}{\recCtx}{P}{d}{D} & \typeRecD{\ctx, c : B}{\recCtx}{P}{d}{D}}
\end{mathpar}

The right rules state the process has to offer either the left type or the right type respectively. The left rule says we need to be prepared to handle either type. It is interesting to observe that the usual problems with unions in functional languages do not arise in our setting. The natural left rule we give here (natural since it is dual to the right rule for intersection) has been shown to be unsound in functional languages \cite{Barbanera95ic}. One somewhat heavy solution limits the left rule to expressions in evaluation position \cite{DunfieldP04}. The straightforward left rule turns out to be already sound here, essentially due to linearity and the use of the sequent calculus.

The usual subtyping rules are given below. The $\union\Right$ rules are derivable by general subtyping, so they don't need to be explicitly added to the system.

\begin{mathpar}
  \infer=[\Sub{\union\Right_1}]{A \sub B_1 \union B_2}
    {A \sub B_1}
  \and \infer=[\Sub{\union\Right_1}]{A \sub B_1 \union B_2}
    {A \sub B_2}
  \and \infer=[\Sub{\union\Left}]{A_1 \union A_2 \sub B}
    {A_1 \sub B & A_2 \sub B}
\end{mathpar}


For an example where unions are the most natural way to express a property, consider moving to a binary representation of natural numbers. We define the type of binary string where the least significant bit is sent first and the string of bits is terminated with $\m{eps}$:
\[
\begin{array}{lcl}
\m{Bits} & = & \internal\{\m{eps} : \terminate, \m{zero} : \m{Bits}, \m{one} : \m{Bits}\}
\end{array}
\]
We can define bit strings in standard form (no leading zeros) as follows:
\[
\begin{array}{lcl}
\m{Std} & = & \m{Empty} \union \m{StdPos} \\
\m{Empty} & = & \internal\{\m{eps} : \terminate\} \\
\m{StdPos} & = & \internal\{\m{one} : \m{Std}, \m{zero} : \m{StdPos}\}
\end{array}
\]
We are able to naturally express standard bit strings as either an empty bit string \emph{or} a positive one; expressing such types without unions would be cumbersome at the very least.%
\footnote{The acute reader might notice that $\m{Std}$ actually fails our simple contractiveness criteria. This does not lead to unsoundness, however, since its one level unfolding is contractive (i.e.\ inlining $\m{Empty}$ and $\m{StdPos}$ makes it contractive). We write it in closed form here for better readability, but one can consider this simple syntactic sugar. It is not too hard to formulate a sound contractiveness condition that allows such definitions, but we decided to err on the side of simplicity for this paper.}
Now we can write an increment function that preserves bit strings in standard form:
\[
\begin{array}{l}
\m{inc} : \m{Std} \lolli \m{Std} \intersect \m{StdPos} \lolli \m{StdPos} \intersect \m{Empty} \lolli \m{StdPos} \\[1ex]
c \leftarrow \m{inc}\; d = \\
\quad \mb{case}\; d\; \mb{of} \\
\qquad \m{eps} \rightarrow \mb{wait}\; d \semi c.\m{one} \semi c.\m{eps} \semi \mb{close}\; c \\
\qquad \m{zero} \rightarrow c.\m{one} \semi c \leftarrow d \\
\qquad \m{one} \rightarrow c.\m{zero} \semi \m{inc}\; d
\end{array}
\]
This example also demonstrates that for a recursively defined type we may need to specify more information than we ultimately care about, since checking this definition just against the type $\m{Std} \lolli \m{Std}$ will fail, and we need to assign the more specific type for the type checking to go through. This is because of the nature of our system which essentially requires the type checker to verify a fixed point rather than infer the least one. This has proven highly beneficial for providing good error messages even in the simpler case of pure subtyping, without intersections and unions~\cite{Griffith16phd}.

\subsection{Subtyping Revisited}

In line with our propositional interpretation of intersections and unions, one would naturally expect the usual properties of these to hold in our system. For example, unions should distribute over intersections and vice versa, that is, the following equalities should be admissible:
\begin{mathpar}
   (A_1 \union B) \intersect (A_2 \union B) \typeeq (A_1 \intersect A_2) \union B \\
   (A_1 \union A_2) \intersect B \typeeq (A_1 \intersect B) \union (A_2 \intersect B)
\end{mathpar}

Going from right to left turns out to be easy, but we quickly run into a problem if we try to do the other direction: whether we break down the union on the right or the intersection on the left, we always lose half the information we need to carry out the rest of the proof.\footnote{This issue does not come up in the other direction since intersection right and union left rules are invertible, that is, they preserve all information.}

Our solution is doing the obvious: if the problem is losing half the information, well, we should just keep it around. This suggests a system where the single type on the left and the type on the right are replaced with \emph{(multi)sets} of types. That is, instead of the judgment $A \le B$, we use a judgment of the form $A_1, \ldots, A_n \subA B_1, \ldots, B_n$, where the left of $\subA$ is interpreted as a conjunction (intersection) and the right is interpreted as a disjunction (union).%
\footnote{We use multisets rather than sets since types have nontrivial equality, so it is not obvious when we should combine them into one.
}
This results in a system reminiscent of \cite{Gentzen35, Girard87}. However, we take a slightly different approach since we are working with coinductive rules.

The rules are given in \cref{subtyping-multi}. We use $\typeList$ and $\typeListB$ to denote multisets of types. The intersection left rules are combined into one rule that keeps both branches around. The same is done with union right rules. Intersection right and union left rules split into two derivations, one for each branch, but keep the rest of the types unchanged. We can unfold a recursive type on the left or on the right. When we choose to apply a structural rule, we have to pick exactly one type on the left and one on the right with the same structure. We conjecture that matching multiple types might give us distributivity of intersection and union over structural types, although a naive extension along these lines fails to satisfy type safety.

\begin{rules}[subtyping-multi]{Subtyping with multiple hypothesis and conclusions; coinductively with
respect to a fixed signature $\eta$}
  \infer=[\SubA{\intersect}\Right]{\typeList \subA \typeListB, A_1 \intersect A_2}
    {\typeList \subA \typeListB, A_1 \and \typeList \subA \typeListB, A_2}
  \and \infer=[\SubA{\intersect}\Left]{\typeList, A_1 \intersect A_2 \subA \typeListB}
    {\typeList, A_1, A_2 \subA \typeListB}
  \\ \infer=[\SubA{\union}\Right]{\typeList \subA \typeListB, A_1 \union A_2}
    {\typeList \subA \typeListB, A_1, A_2}
  \and \infer=[\SubA{\union}\Left]{\typeList, A_1 \union A_2 \subA \typeListB}
    {\typeList, A_1 \subA \typeListB & \typeList, A_2 \subA \typeListB}
  \\ \infer=[\SubA{\terminate}]{\typeList, \terminate \subA \typeListB, \terminate}{}
  \and \infer=[\SubA\tensor]{\typeList, A \tensor B \subA \typeListB, A' \tensor B'}
    {A \subA A' & B \subA B'}
  \and \infer=[\SubA\internal]{\typeList, \internals A I \subA \typeListB, \internals {A'} J}
    { I \subseteq J
    & A_\indexVar \subA A_\indexVar'~\text{for}~\indexVar\in I
    }
  \and \infer=[\SubA\lolli]{\typeList, A \lolli B \subA \typeListB, A' \lolli B'}
    {A' \subA A & B \subA B'}
  \and \infer=[\SubA\external]{\typeList, \externals A I \subA \typeListB, \externals {A'} J}
    { J \subseteq I
    & A_\indexVar \subA A_\indexVar'~\text{for}~\indexVar\in J
    }
  \\ \infer=[\SubA{\m{Def}\Right}]{\typeList \subA \typeListB, t}
     {(t = A \in \eta) & \typeList \subA \typeListB, A}
  \and \infer=[\SubA{\m{Def}\Left}]{\typeList, t \subA \typeListB}
     {(t = A \in \eta) & \typeList, A \subA \typeListB}
\end{rules}

\subsection{Reinterpreting Choice}
\label{reinterpreting-choice}

In this section, we show that that intersections and unions are useful beyond their refinement interpretation, and help us understand external and internal choices better. Take external choice, for instance. A comparison between the typing rules for intersections and external choice reveal striking similarities. The only difference, in fact, is that internal choice has process-level constructs whereas intersections are implicit.

Consider the special case of binary external choice: $\external\braces{\m{inl} : A, \m{inr} : B}$. This type says: I will act as $A$ if you send me $\m{inl}$ \emph{and} I will act as $B$ if you send me $\m{inr}$. We know the \emph{and} can be interpreted as an intersection, and either side can be thought of as a singleton internal choice. A similar argument can be given for internal choice and unions. This gives us the following redefinitions of $n$-ary external and internal choices:
\begin{mathpar}
  \externals A I \defined \bigintersect_{\indexVar \in I}{\external\braces{\lab_\indexVar : A_\indexVar}} \\
  \internals A I \defined \bigunion_{\indexVar \in I}{\internal\braces{\lab_\indexVar : A_\indexVar}}
\end{mathpar}
It is a straightforward calculation that these definitions satisfy the typing and subtyping rules for external and internal choices.

\section{Algorithmic System}
\label{algorithmic}

In this section, we show that subtyping and type-checking are decidable by designing an algorithm that takes in a (sub)typing judgment and produces true if and only if there is a derivation. Note that everything in the judgment is considered an input.

\subsection{Algorithmic Subtyping}

The subtyping judgment we gave is already mostly algorithmic (a necessity of working with coinductive rules), so we only have to tie up a few loose ends. The first is deciding which rule to pick when multiple are applicable. We apply $\SubA{\intersect\Right}$, $\SubA{\intersect\Left}$, $\SubA{\union\Right}$, $\SubA{\union\Left}$, $\SubA{\m{Def}\Right}$, $\SubA{\m{Def}\Left}$ eagerly since these are invertible. At some point, all types must be structural (since definitions are restricted to be contractive), at which point we non-deterministically pick a structural rule and continue. In the implementation, we backtrack over these choices.

Second, the coinductive nature of subtyping means we can (and often will) have infinite derivations. We combat this by using a cyclicity check (similar to the one in \cite{GayH05}): we maintain a context of previously seen subtyping comparisons and immediately terminate with success if we ever compare the same pair of sets of types again. Every recursive step corresponds to a rule, which ensures a productive derivation. We know there cannot be an infinite chain of new types due to the contractiveness restriction which implies an upper-bound on the size of the previously-seen set. A more formal treatment can be found in \cite{Stone05un}.

\subsection{Algorithmic Type-checking}

Designing a type checking algorithm is quite simple for the base system where we only have structural types (no recursive types or subtyping), since the form of the process determines a unique applicable typing rule. The $\cut$ rule causes a small problem since we do not have a type for the new channel. This is solved by requiring a type annotation when necessary such that the new term becomes $\tspawnType c {P_c} A {Q_c}$. The overwhelmingly common case where it is \emph{not} necessary is when $P_c$ is a defined process name $X$ because we simply fall back on its given type.

In the extended system with subtyping and property types, type-checking is trickier for two reasons: (1) subsumption can be applied anytime where one of the types in $A \le B$ can be anything (the other will be fixed due to typing rules, but one causes enough damage), and (2) intersection left and union right rules lose information which means they have to be applied non-deterministically. The latter issue is resolved by switching to a judgment where each channel (whether on the left or the right) is assigned a multiset of types.  These multisets are interpreted conjunctively for channels used (on the left) and disjunctively for the channel provided (on the right). This makes intersection left and union right rules invertible, so they can be applied eagerly.

The former problem is solved by checking subtyping only at the identity rule (forwarding). This relies on the subformula property for the sequent calculus, excepting only the cut rule which is annotated. The new judgment is written $\typeRecAJ P c \typeList$, where $\ctx$ and $\typeList$ are multisets. Typing rules are given in \cref{algorithmic-typing}. Note the explicit $\intersect\Left$, $\union\Right$, $\m{Def}\Left$, and $\m{Def}\Right$. These rules were derivable in the declarative system using subsumption, which is no longer possible since application of subsumption is restricted to forwarding processes.

\begin{rules}[algorithmic-typing]{Process Typing in the Algorithmic System}
  \infer[\intersect\Right]{\typeRecAJR \ctx P c {A \intersect B, \typeList}}
    { \typeRecAJR \ctx P c {A, \typeList}
    & \typeRecAJR \ctx P c {B, \typeList}
    }
  \and \infer[\intersect\Left]{\typeRecAJR{\ctx, c : (\typeList, A \intersect B)}{P}{d}{\typeListB}}
    {\typeRecAJR{\ctx, c : (\typeList, A, B)}{P}{d}{\typeListB}}
  \\ \infer[\union\Right]{\typeRecAJR \ctx P c {A \union B, \typeList}}
    {\typeRecAJR \ctx P c {A, B, \typeList}}
  \and \infer[\union\Left]{\typeRecAJR{\ctx, c : (\typeList, A \union B)} P d \typeListB}
    { \typeRecAJR {\ctx, c : (\typeList, A)} P d \typeListB
    & \typeRecAJR {\ctx, c : (\typeList, B)} P d \typeListB
    }
  \\ \infer[\m{Def}\Right]{\typeRecAJ P c {t, \typeList}}
    { (t = A \in \eta) & \typeRecAJ P c {A, \typeList} }
  \and \infer[\m{Def}\Left]{\typeRecAJR {\ctx, c : (\typeList, t)} P d {\typeListB}}
    { (t = A \in \eta) & \typeRecAJR {\ctx, c : (\typeList, A)} P d {\typeListB} }
  \\ \infer[\id]{ \typeRecAJR {c : \typeList} {\tfwd d c} {d} {\typeListB} }
    { \typeList \subA \typeListB }
  \and \infer[\cut]{ \typeRecAJR {\ctx, \ctx'} {\tspawnType c {P_c} A {Q_c}} {d} \typeList }
    { \typeRecAJR \ctx {P_c} {c} {A}
    & \typeRecAJR {\ctx', c : A} {Q_c} {d} {\typeList}
    }
  \\ \infer[\terminate\Right]{\typeRecAJR{\emptyCtx}{\tclose c}{c}{\terminate, \typeList}}
   {}
  \and \infer[\terminate\Left]{\typeRecAJR{\ctx, c : (\typeList, \terminate)}{\twait c P} d \typeListB}
    { \typeRecAJR {\ctx} P d \typeListB}
  \and \infer[\tensor\Right]{\typeRecAJR{\ctx, \ctx'}{\tsend c d {P_d} Q }{c}{A \tensor B, \typeList}}
    { \typeRecAJR \ctx P d A
    & \typeRecAJR {\ctx'} Q c B
    }
  \and \infer[\tensor\Left]{ \typeRecAJR{\ctx, c : (\typeList, A \tensor B)}{\trecv{d}{c}{P_d}}{e}{\typeListB} }
    { \typeRecAJR{\ctx, d : A, c : B}{P_d}{e}{\typeListB} }
  \and \infer[\internal\Right]{\typeRecAJR \ctx { \tselect{c}{i}{P} } {c} {\internals{A}{I}, \typeList }}
    { i \in I
    & \typeRecAJR \ctx {P}{c}{A_i}
    }
  \and \infer[\internal\Left]{ \typeRecAJR { \ctx, c : (\typeList, \internals{A}{I}) } { \tcase{c}{\tbranches{P}{J}} } {d} {\typeListB} }
   { I \subseteq J
   & \typeRecAJR{\ctx, c : A_k}{P_k}{d}{\typeListB}~\text{for}~k\in I
   }
  \and \infer[\lolli\Right]{ \typeRecAJR{\ctx}{\trecv{d}{c}{P_d}}{c}{A \lolli B, \typeList} }
    { \typeRecAJR{\ctx, d : A}{P_d}{c}{B} }
  \and \infer[\lolli\Left]{\typeRecAJR{\ctx, \ctx', c : (\typeList, A \lolli B)}{ \tsend{c}{d}{P_d}{Q} } {e}{\typeListB}}
    { \typeRecAJR{\ctx}{P_d}{d}{A}
    & \typeRecAJR{\ctx', c : B}{Q}{e}{\typeListB}
    }
  \and \infer[\external\Right]{ \typeRecAJR \ctx { \tcase{c}{\tbranches{P}{I}} } {c} {\externals{A}{J}, \typeList} }
   { J \subseteq I
   & \typeRecAJR \ctx {P_k}{c}{A_k}~\text{for}~k\in J
   }
  \and \infer[\external\Left]{\typeRecAJR{\ctx, c : (\typeList, \externals{A}{I})} { \tselect{c}{i}{P} } {d} {\typeListB}}
    { i \in I
    & \typeRecAJR{\ctx, c : A_i}{P}{d}{\typeListB}
    }
  \and \infer[\m{def}]{ \typeRecAJ {\emptyset} {X} {d}{A} }
    { X = P :: (c : A) \in \recCtx }

\end{rules}

\subsection{Equivalence to the Declarative System}

Next, we show that the algorithmic system is sound and complete with respect to the declarative system (modulo erasure of type annotations, which we denote by $\erase P$). Due to space limitations, we can only give very brief proof sketches here. Interested readers are referred to the first author's thesis \cite{Acay16}.  We define $\any\typeList$ as the union of all the types in $\typeList$, and similarly for $\all\typeList$.  For contexts
we define $\all(c_1{:}\typeList_1, \ldots, c_n{:}\typeList_n) = c_1{:}\all\typeList_1, \ldots, c_n{:}\all\typeList_n$.

\begin{theorem}[Soundness of Algorithmic Typing]
  If $\typeRecAJ P c \typeList$, then $\typeRecD {\all\ctx} {\recCtx} {\erase P} c {\any \typeList}$.
\end{theorem}
\begin{proof}
  By induction on the typing derivation. The only non-straightforward cases are $\intersect\Right$ and $\union\Left$, which depend on the distributivity of intersection and union over each other (which is one of the reasons why we insisted such be the case while designing the subtyping relation).
\end{proof}

\begin{lemma}[Completeness of Delayed Subtyping]
  \label{algorithmic:delegation-sub}
  The following are admissible:
  \begin{itemize}
    \item If $\typeRecAJ P c \typeList$ and $\any\typeList \subA \typeListB$ then $\typeRecAJ P c {\typeListB}$.
    \item If $\typeRecAJR {\ctx, d : \typeList} P c \typeListB$ and $\typeList' \subA \all\typeList$ then $\typeRecAJR {\ctx, d : \typeList'} P c \typeListB$.
  \end{itemize}
  Note that the type annotations in $P$ stay the same.
\end{lemma}
\begin{proof}
  By lexicographic induction, first on the structure of $P$, then on the combined sizes of involved types.
\end{proof}

\begin{theorem}[Completeness of Algorithmic Typing]
  If $\typeRecDJ P c A$, then there exists $P'$ such that $\erase{P'} = P$ and $\typeRecAJ {P'} c A$.
\end{theorem}
\begin{proof}
  By induction on the typing derivation, using \cref{algorithmic:delegation-sub} for $\irb{Sub}\Right$ and $\irb{Sub}\Left$.
\end{proof}

\section{Metatheory}
\label{metatheory}

Our main contribution is proving type safety for the system with intersections and unions, which we do so by showing the standard progress and preservation theorems, renamed to deadlock freedom and session fidelity, respectively, within this context. Since the algorithmic system is more well behaved (no subsumption), we use the algorithmic judgment in the statements and proofs of these results. Type safety for the declarative system follows from its equivalence to the algorithmic system. We only state the theorems here. Full proofs can be found in~\cite{Acay16}.

In a functional setting, progress states a well-typed expression either takes a step or is a value. The corresponding notion of a value is a \emph{poised} configuration. A configuration is poised if every process in it is, and a process is poised if it is waiting to communicate with its client. With this definition, we can state the progress theorem:

\begin{theorem}[Progress]
If $\providesCtx \config \ctx$ then either
\begin{enumerate}
  \item $\steps{\config}{\config'}$ for some $\config'$, or
  \item $\config$ is poised.
\end{enumerate}
\end{theorem}
\begin{proof}
  By induction on $\providesCtx \config \ctx$ followed by a nested induction on the typing of the root process for the $\confOne$ case. When two processes are involved, we also need inversion on client's typing.
\end{proof}

\begin{theorem}[Preservation]
If $\providesCtx \config \ctx$ and $\steps{\config}{\config'}$ then $\providesCtx {\config'} \ctx$.
\end{theorem}
\begin{proof}
  By inversion on $\steps{\config}{\config'}$, followed by induction on the typing judgments of the involved processes.
\end{proof}

\section{Related Work and Conclusion}
\label{conclusion}

Padovani describes a calculus similar to ours~\cite{Padovani10} where he interprets internal and external choices as union and intersection, respectively. This resembles what we did in \cref{reinterpreting-choice} except we keep singleton choices at the type and term levels to maintain the connection to linear logic, whereas Padovani is able to remove them completely since his calculus is based on \textsc{CSS}~\cite{NicolaH87}. While we give axiomatic rules for (sub)typing, he takes a semantic approach, which we believe is complementary to our work. However, semantic definitions make deriving algorithmic rules harder and he leaves this as future work. More significantly, Padovani does not consider higher-order types and processes, which means it is not possible to communicate channels along channels. Moreover, his calculus only deals with the interaction between two processes that are required to have matching (or dual) types and behaviors (for example, if one sends, the other must receive etc.). We consider a tree of processes where each process can use many providers (as long as it respect the behavior along their channels) and even spawn new ones. For example, a client using two providers could communicate with one of them while the other is idle, or ignore both altogether and only communicate with \emph{its} client.

Castagna et al.\ describe a generic framework~\cite{CastagnaDGP09} that make use of set operations (intersection, union, and negation) for sessions and take a semantic approach as well. Their framework has the advantage that it is agnostic to the underlying functional language. There are descriptions of our base calculus that take a similar approach~\cite{ToninhoCP13} and it should not be too hard to extend them to cover our contributions. Contrary to Padovani, Castagna et al.\ have higher order sessions and give algorithms to decide all semantic relations they describe. However, their system and presentation are significantly different from ours because of their semantic emphasis, inclusion of negation (which makes their session language Boolean), and treatment of process composition (which is closer to Padovani's).  In particular, they do not have a general primitive for spawning new processes or forwarding.

We introduced intersections and unions to a simple system of session types, and demonstrated how they can be used to refine behavioral specifications of processes. Some aspects that would be important in a full accounting of the system are omitted for the sake of brevity or are left as future work. For example, integrating an underlying functional language \cite{ToninhoCP13}, adding shared channels \cite{CairesP10,PfenningG15}, or considering asynchronous communication \cite{DeYoungCPT12,PfenningG15,Griffith16phd} are straightforward extensions based on prior work. In addition, it would be very useful to have behavioral polymorphism \cite{CairesPPT13} and abstract types. Their interaction with subtyping, intersections, and unions is an interesting avenue for future work.

\paragraph{Acknowledgments.}
This work was funded in part by NSF grant CNS1423168 and by the FCT (Portuguese Foundation for Science and Technology) through the Carnegie Mellon Portugal Program.  We would like to thank the anonymous referees for their many helpful suggestions.


\bibliographystyle{eptcs}
\bibliography{bibliography,db}


\end{document}